\documentclass[reqno,a4paper,11pt]{amsart}
\usepackage{amsmath,amstext,amsfonts,amsbsy,eucal,amssymb}
\usepackage[latin1]{inputenc}

\numberwithin{equation}{section}

\newtheorem{theorem}{Theorem}[section]
\newtheorem{lemma}[theorem]{Lemma}

\newtheorem{proposition}[theorem]{Proposition}

\newtheorem{conjecture}[theorem]{Conjecture}

\theoremstyle{definition}
\newtheorem{definition}[theorem]{Definition}
\newtheorem{example}[theorem]{Example}
\newtheorem{remark}[theorem]{Remark}

\begin{document}

\parskip 4pt
\baselineskip 16pt

%%%%%%%%%%%%%%%%%%%%%%%%%%%%%%%%%
%%%%%%%%%%%%%%%%%%%%%%%%%%%%%%%%%

\title[Invariants for a family of discrete equations]
{Invariants for a family of discrete equations with the Laurent property}

\author{Andrei K. Svinin}
\email{svinin@icc.ru}
\address{Matrosov Institute for System Dynamics and Control Theory, Siberian Branch of Russian Academy of Sciences, PO Box 292, 664033 Irkutsk, Russia}
\date{\today}

\begin{abstract}
Recently, we have found an infinite family of  homogeneous discrete equations of odd order possessing the Laurent property. The first representative of this family is the well-known Somos-5 equation, which under certain conditions generates the integer sequence A006721, which has numerous applications.

In this work, we construct a finite set of independent invariants for our equations. We show, through examples, that the presence of these invariants allows us to find a more general criterion for the integrality of sequences compared to what the usual Laurent property provides.
\end{abstract}

\maketitle

\section*{Introduction}
\label{sec1}

In 1989, Michael Somos, in his research on elliptic functions, discovered several nonlinear recurrence relations which, despite their nonlinearity, generate integer sequences for suitable initial data. One of them, which is directly relevant to the present work, is the bilinear recurrence relation
\[
t_nt_{n+5}= t_{n+1}t_{n+4}+ t_{n+2}t_{n+3},
\]
now known as Somos-5. Somos posed the problem of proving -- and, even better, explaining -- this integrality. Specifically for the Somos-5 equation, one needed to prove the integrality of the sequence defined by the initial conditions $t_j=1$ for $j=0,1,2,3,4$.
In the OEIS online encyclopedia, this sequence is identified as \textrm{A006721}. It later turned out that this sequence has numerous diverse applications, listing which would require a separate survey. We give key references here. In \cite{Buchholz}, a connection was found between the sequence \textrm{A006721} and an infinite family of Heronian triangles with two rational medians. The same paper mentions a relation of this sequence to an elliptic curve and related special functions. In \cite{Eager, Speyer}, one can find combinatorial interpretations of the terms of the sequence \textrm{A006721}. This sequence is also of interest from the viewpoint of number theory \cite{Davis, Robinson}.

Somos's discovery sparked a wave of research, since the nature of the integrality of these sequences remained unclear for some time. David Gale, in his article \cite{Gale}, gave a survey of the state of affairs at that time. One of his observations was that there actually exist quite many such polynomial recurrence relations and several examples had been found. In particular, he mentioned an infinite family of recurrence relations of Raphael Robinson:
\[
t_nt_{n+N}= t_{n+1}t_{n+N-1}+ t_{n+2}t_{n+N-2},\;\; N\geq 4.
\]
Gale, in turn, conjectured that the recurrence relation
\begin{equation}
t_{n}t_{n+N}=\alpha t_{n+a}t_{n+N-a}+\beta t_{n+c}t_{n+N-c},
\label{6753098}
\end{equation}
for any $N\geq 4$ and $1\leq a<c\leq \left\lfloor N/2 \right\rfloor$, with arbitrary integer coefficients $\alpha$ and $\beta$ and unit initial conditions, also generates an integer sequence. Nowadays, this relation is called the Gale--Robinson equation. More precisely, it is an infinite family of equations, where $N$ is the order of the corresponding equation.

Subsequently, the Gale--Robinson equations gained enormous popularity due to applications in combinatorics and theoretical physics. There currently exists an extensive amount of literature devoted to the combinatorial meaning of the integer sequences generated by Gale--Robinson equations, including connections to theoretical physics \cite{Eager}. Perhaps the most intuitive interpretation of the sequences generated by Gale--Robinson equations with coefficients $\left(\alpha, \beta\right)$ equal to unity was found in the works \cite{Bousquet,Speyer}.
\begin{definition}
We say that (\ref{6753098}) is a Gale--Robinson equation of type $\left(N, a, c\right)$.
\end{definition}

A major breakthrough in the problem of explaining the integrality of sequences generated by certain polynomial recurrence relations was achieved by Andrei Zelevinsky and Sergey Fomin. In the early 2000s, they were developing their theory of cluster algebras. In \cite{Fomin}, using this theory, they proved, as a special case, the Gale--Robinson conjecture \cite{Fomin}.
The essence is that these equations have the Laurent property, which Fomin and Zelevinsky discovered in their developing theory of cluster algebras, and the Gale--Robinson equations with $\left(\alpha, \beta\right)=\left(1, 1\right)$ represent recurrence relations naturally arising in the theory of cluster algebras.

Consider the recurrence relation
\begin{equation}
t_nt_{n+N}=L\left(t_{n+1},\ldots, t_{n+N-1}\right),
\label{88755497}
\end{equation}
where it is assumed that $L$ is some (Laurent) polynomial in its variables.
\begin{definition}  \label{094533}
Equation (\ref{88755497}) has the Laurent property if $t_n\in\mathbb{Z}[t_0^{\pm 1},\ldots, t_{N-1}^{\pm 1}],\;\; \forall n\in\mathbb{Z}$.
\end{definition}
\begin{remark} 
In the literature, for virtually all equations with the Laurent property of the form (\ref{88755497}), the right-hand side $L$ is a polynomial in its variables.
\end{remark} 

As an elementary consequence of the Laurent property, we obtain that if $t_j=\pm 1$ for $j=0,\ldots, N-1$, then the corresponding recurrence relation generates some integer sequence. This condition is sufficient for integrality of the sequence, though it is not necessary.

Despite the fact that the condition in Definition \ref{094533} looks extremely restrictive, it turns out that there exist quite many recurrence relations with the Laurent property. Such equations either follow from the theory of cluster algebras, as for example in \cite{Fordy}, or within the framework of the approach of Lam and Pylyavskyy \cite{Lam}, which generalizes the theory of cluster algebras. Both approaches are, in essence, combinatorial.

In a recent paper \cite{Svinin7}, we presented an infinite family of homogeneous recurrence relations of the form
\begin{equation}
t_nt_{n+2g+3}=\frac{\sum_{j=0}^{g}\alpha_j B^j_{2g-j}(n+1)}{\prod_{j=3}^{2g} t_{n+j}},\; g\geq 1, 
\label{8878886500097}
\end{equation}
for which, as we proved, the Laurent property holds. Note that relation (\ref{8878886500097}) represents an equation of the form (\ref{88755497}) with parameters $\alpha_j$, which are considered arbitrary.
The right-hand side of (\ref{8878886500097}) is defined by certain homogeneous discrete polynomials $B^k_s(n)$ in $t_n$.\footnote{We regularly use the term discrete polynomial. In this context it means that $B^k_s(n)$ is some polynomial in a finite number of variables $t_n, t_{n+1},\ldots$} These polynomials are defined by means of a recurrence relation. We will write this relation in the next section; for now we only note that $B^j_{2g-j}(n)$ is a homogeneous polynomial in the variables $\left(t_n,\ldots, t_{n+2g+1}\right)$. The degree of this polynomial is $2g$. Thus, relation (\ref{8878886500097}), for any $g\geq 1$, is a homogeneous equation. Note that the right-hand side of our equations of the form (\ref{8878886500097}) is, generally speaking, not a polynomial, but a Laurent polynomial.

For clarity, we write out the first two representatives of this family. For $g=1$, the product in the denominator on the right-hand side of (\ref{8878886500097}) should be taken as unity. In this case, (\ref{8878886500097}) turns into the Somos-5 equation with parameters, namely\footnote{Hereafter we call this equation simply the Somos-5 equation.}
\begin{equation}
t_nt_{n+5}=  \alpha_0 t_{n+2}t_{n+3}+\alpha_1 t_{n+1}t_{n+4}.
\label{5}
\end{equation}
The next representative of this family is the equation
\begin{equation}
t_nt_{n+7}=\alpha_0 t_{n+2}t_{n+5}+\alpha_1 \frac{t_{n+1}t_{n+4}^2t_{n+5}+t_{n+2}^2t_{n+5}^2+t_{n+2}t_{n+3}^2t_{n+6}}{t_{n+3}t_{n+4}}+\alpha_2 t_{n+1}t_{n+6}.
\label{7}
\end{equation}
As already noted, Somos-5, like all Gale--Robinson equations (\ref{6753098}), fits well into the theory of cluster algebras \cite{Fordy}. However, for our equations of the form (\ref{8878886500097}), such a combinatorial interpretation is not yet known.

Let us highlight a key fact about our equations of the form (\ref{8878886500097}), which can be viewed, on the one hand, as a property of these equations, and on the other hand, as their definition.
Let $\left( u_n \right)_{n\in\mathbb{Z}}$ be the unknown sequence. For each $g\geq 1$, we define the associated recurrence relation
\begin{equation}
\prod_{j=0}^{2g} u_{n+j}=\sum_{j=0}^g \alpha_j T^j_{2g-j}(n+1),
\label{6766438}
\end{equation}
where $T^k_s(n)$ are certain discrete homogeneous polynomials in $u_n$, which will be defined in the next section. The fact is that by means of the substitution
\begin{equation}
u_n=\frac{t_nt_{n+3}}{t_{n+1}t_{n+2}}
\label{676678888768}
\end{equation}
the relation (\ref{6766438}) reduces to the homogeneous form (\ref{8878886500097}).
In what follows, for convenience, for any $g\geq 1$, we will denote the corresponding equation by $R_{2g+3}$.
\begin{definition}
We call an $R_{2g+3}$-sequence any numerical sequence $\left(t_n\right)_{n\in\mathbb{Z}}$ satisfying $R_{2g+3}$ with some data:
\[
\left(t_0,\ldots, t_{2g+2}; \alpha_0,\ldots, \alpha_g\right)\in\mathbb{C}^{3g+4}.
\]
\end{definition}

We list some properties of the equations in our family. In principle, all these properties are listed and proved in \cite{Svinin7}.
\begin{proposition} \label{0027}
For any $g\geq 1$, $R_{2g+3}$ admits a reduction to the Gale--Robinson equation (\ref{6753098}) of type $\left(N, a, c\right)=\left(2g+3, 1, 2\right)$.
\end{proposition} 
\begin{proposition}   \label{00444327}
Let $\left( t_n \right)_{n\in\mathbb{Z}}$ be an arbitrary $R_{2g+3}$-sequence. The transformation $t_n\mapsto AB^n t_n$, where $A$ and $B$ are arbitrary nonzero constants, yields another $R_{2g+3}$-sequence.
\end{proposition}
\begin{proposition}   \label{00098767}
The number of monomials defining $R_{2g+3}$ is $F_{2g+1}+1$, where $F_k$ denotes the $k$-th Fibonacci number.
\end{proposition}
\begin{proposition}   \label{004432767}
Set $n=0$ in (\ref{8878886500097}). The right-hand side $L(t_{1},\ldots, t_{2g+2})$ of this relation, which is generally a Laurent polynomial, has the symmetry $L(t_{1},\ldots, t_{2g+2}) = L(t_{2g+2},\ldots, t_{1})$.
\end{proposition}
The main property of $R_{2g+3}$ is expressed by the following theorem \cite{Svinin7}.
\begin{theorem}  \label{76888766}
The equation $R_{2g+3}$ for any $g\geq 1$ has the Laurent property. More precisely, for all $n\in\mathbb{Z}$, $t_n$ belongs to the ring $\mathbb{Z}[t_0^{\pm 1},\ldots, t_{2g+2}^{\pm 1}; \alpha_0,\ldots,\alpha_g]$.
\end{theorem}
Another property of the equation $R_{2g+3}$, as shown in \cite{Svinin7}, is that it can be embedded into a discrete Mumford dynamical system defined in \cite{Hone2}. This dynamical system is defined on a phase space of dimension $3g+1$. As shown in \cite{Hone2}, it is Liouville integrable. This system is defined by a Lax representation from which, in principle, a certain finite set of invariants can be derived. However, extracting explicit expressions for the invariants from the Lax representation is technically difficult.
The purpose of this paper is to construct, for each equation $R_{2g+3}$, a finite set of independent invariants $\left(H_1,\ldots, H_g\right)$ in explicit form, independently of the Lax representation.

In the next section we show how, precisely, our equations $R_{2g+3}$ are defined. For this we essentially need to define all the necessary discrete polynomials $B^k_s(n)$. Section \ref{87650987} is devoted to constructing invariants for the equations $R_{2g+3}$.
In Section \ref{765498} we discuss a relation conjecturally equivalent to the associated equation (\ref{6766438}). For small $g$, this equivalence is verified directly. The equivalent equation contains $\left(H_1,\ldots, H_g\right)$ and $\alpha_g$ as parameters. The interest in it is that by substituting (\ref{676678888768}) into this equation we obtain a relation equivalent to $R_{2g+3}$ which has its own peculiarities from the viewpoint of the Laurent property.
The proven statements regarding an analogue of the Laurent property allow us to present more general criteria for integrality of $R_{2g+3}$ sequences than the one given by the ordinary Laurent property. We discuss in detail the cases $g=1$ and $g=2$.
\section{Discrete polynomials $T^k_s(n)$ and $B^k_s(n)$}

\subsection{Discrete polynomials $T^k_s(n)$}

To define our equations $R_{2g+3}$, we need to define the discrete polynomials appearing in the right-hand side of (\ref{8878886500097}), and for this, in turn, we need to define the discrete polynomials $T^k_s(n)$. In a somewhat modified form, these polynomials and their generalizations were used in \cite{Svinin2}. From a practical point of view, a good definition of $T^k_s(n)$ is the recurrence relation:
\begin{equation}
T^k_s(n)=\sum_{j=0}^{s-k} u_{n+j} T^{k-1}_{s-j-1}(n+j+2),\;\; \forall s\geq k.
\label{676678}
\end{equation}
For example, we can start with $T^0_s(n)=1$ and compute the polynomials $T^k_s(n)$ step by step for increasing $k$, but we can also write $T^k_s(n)$ in explicit form:
\begin{equation}
T^k_s(n)=\sum_{0\leq \lambda_1<\cdots< \lambda_k\leq s-1} \prod_{j=1}^k u_{n+\lambda_j+j-1}.
\label{67999998}
\end{equation}
Expression (\ref{67999998}) is a summation over a finite set of admissible values $\lambda_j$. Denote $D_{k, s}=\left\{\lambda_j : 0\leq \lambda_1<\cdots< \lambda_k\leq s-1\right\}$. Note that for $s=1,\ldots, k-1$ we have $D_{k, s}=\varnothing$, and consequently $T^k_s(n)=0$ for these values of $s$. On the other hand, if $s=k$, then the set $D_{k, s}$ consists of one element $\left\{\lambda_j=j,\;\; j=0,\ldots, k-1\right\}$ and thus we obtain
\begin{equation}
T^k_k(n)=\prod_{j=0}^{k-1}u_{n+2j},\;\; \forall k\geq 1.
\label{6799999000988}
\end{equation}

A disadvantage of the explicit formula (\ref{67999998}) is that it is rather difficult to use in practice compared to the recurrence relation (\ref{676678}). Nevertheless, some useful facts can be extracted from it. Therefore it is better to use these definitions together. From formula (\ref{67999998}) we see that $T^k_s(n)$ is a homogeneous polynomial in the variables $\left(u_n,\ldots, u_{n+k+s-2}\right)$ of degree $k$.

The first thing to show is how to pass from the associated equation (\ref{6766438}) to the relation (\ref{8878886500097}). For this we need to define the map $T^k_s(n)\mapsto B^k_s(n)$. Write
\begin{equation}
B^k_s(n)=T^k_s(n)\prod_{j=1}^{k+s} t_{n+j},\;\; \forall s\geq k,\; k\geq 1,
\label{65400000987}
\end{equation}
provided that we have made the substitution (\ref{676678888768}) in $T^k_s(n)$. But we must be sure that $B^k_s(n)$ are indeed discrete polynomials in $t_n$. The following lemma is useful for proving this fact.
\begin{lemma} \cite{Svinin4}    \label{65439876} 
The two relations
\begin{eqnarray}
T^k_s(n)&=&T^k_{s-1}(n)+u_{n+s+k-2}T^{k-1}_{s-1}(n) \label{6540987} \\
&=&T^k_{s-1}(n+1)+u_nT^{k-1}_{s-1}(n+2)
\label{67776587}
\end{eqnarray}
are identities.
\end{lemma}
\begin{proof}
Relations (\ref{6540987}) and (\ref{67776587}) arise from a suitable partition of the set $D_{k, s}$ into two disjoint subsets. For instance, one can check that (\ref{6540987}) corresponds to the partition $D_{k, s}=D^{(1)}_{k, s}\sqcup D^{(2)}_{k, s}$ with
\[
D^{(1)}_{k, s}=\left\{\lambda_j : \lambda_1=0,\;\; 1\leq \lambda_2<\cdots< \lambda_k\leq s-1\right\}
\]
and 
\[
D^{(2)}_{k, s}=\left\{\lambda_j : 1\leq \lambda_1<\cdots< \lambda_k\leq s-1\right\}.
\]
\end{proof}
\begin{remark}
Note that (\ref{676678}) can be obtained by successive application of (\ref{67776587}). This implies the equivalence of the two different definitions (\ref{676678}) and (\ref{67999998}) for the polynomials $T^k_s(n)$.
\end{remark}
From Lemma \ref{65439876} together with (\ref{65400000987}) we obtain the following statement.
\begin{lemma} \label{76549076}
The discrete polynomials $B^k_s(n)$ can be computed using one of the following two recurrence relations:
\begin{eqnarray}
B^k_s(n)&=&t_{n+k+s}B^k_{s-1}(n)+t_{n+k+s-2}t_{n+k+s+1}B^{k-1}_{s-1}(n) \label{654443287} \\
&=&t_{n+1}B^k_{s-1}(n+1)+t_nt_{n+3}B^{k-1}_{s-1}(n+2).
\label{65445554337}
\end{eqnarray}
\end{lemma}
However, for these recurrence relations to work, additional data are needed. Using (\ref{6799999000988}), we can write:
\[
B^k_k(n)=\prod_{j=0}^{k-1} t_{n+2j} \prod_{j=0}^{k-1} t_{n+2j+3},\;\; \forall k\geq 1.
\]
This polynomial can be computed without using any recurrence relations. We also need to define
\[
B^0_s(n)=\prod_{j=1}^s t_{n+j}, \forall s\geq 1.
\] 
Note that this is consistent with (\ref{65400000987}) provided $T^0_s(n)=1$. Now we have all the data needed to use relation (\ref{654443287}) (or (\ref{65445554337})) for computing the polynomials $B^k_s(n)$, first for $k=1$, then for $k=2$, and so on.

Clearly, $\deg\left(B^0_s\right)=s$ and $\deg\left(B^k_k\right)=2k$. Moreover, the two terms on the right-hand sides of (\ref{654443287}) and (\ref{65445554337}) have the same degree, equal to $k+s$. All this proves that $B^k_s(n)$, computed by (\ref{654443287}) (or (\ref{65445554337})), is indeed a homogeneous polynomial of degree $k+s$ in the variables $\left(t_n,\ldots, t_{n+k+s+1}\right)$.

As a result of substituting (\ref{676678888768}) into the associated recurrence relation (\ref{6766438}), taking into account (\ref{65400000987}), we obtain
\[
\prod_{j=3}^{2g} t_{n+j}\cdot t_nt_{n+2g+3}=\sum_{j=0}^g \alpha_j B^j_{2g-j}(n+1).
\]
Thus, we have fully defined our equations $R_{2g+3}$.
Since 
\[
B^0_{2g}(n+1)=\prod_{j=3}^{2g} t_{n+j}\cdot t_{n+2}t_{n+2g+1}\;\mbox{and}\;\; B^g_g(n+1)=\prod_{j=3}^{2g} t_{n+j}\cdot t_{n+1}t_{n+2g+2},
\]
we can write $R_{2g+3}$ in the form
\[
t_nt_{n+2g+3}=\alpha_0  t_{n+2}t_{n+2g+1}+\frac{\sum_{j=1}^{g-1}\alpha_j B^j_{2g-j}(n+1)}{\prod_{j=3}^{2g} t_{n+j}}+\alpha_g  t_{n+1}t_{n+2g+2}.
\]
This implies Proposition \ref{0027}.
\section{Invariants for the associated recurrence relation (\ref{6766438})}  \label{87650987}

The goal of this section is to construct a finite set of independent invariants for $R_{2g+3}$.

\subsection{Equivalent recurrence relation}

Next we wish to present a finite set of independent invariants for the associated equation (\ref{6766438}). However, for certain reasons, it is more convenient for us to deal with the discrete equation
\begin{equation}
u_{n+2g+1}=u_n\frac{\hat{T}^g_g(n+2)}{\hat{T}^g_g(n+1)},
\label{87654999}
\end{equation}
which, as we will show, is equivalent to (\ref{6766438}).
Here we use the notation
\begin{equation}
\hat{T}^s_s(n)=\sum_{j=0}^s c_j T^{s-j}_{s+j}(n),\;\;\; \forall s\geq 1,
\label{56477777399872}
\end{equation}
where it is assumed that $\left(c_1,\ldots, c_g\right)$ is an arbitrary set of parameters, while $c_0=1$.
\begin{remark}
The one-parameter class of equations of the form (\ref{87654999}) represents a subset of the two-parameter class of equations of the form
\begin{equation}
u_{n+k+s}=u_n\frac{\hat{T}^k_{s-1}(n+2)}{\hat{T}^k_{s-1}(n+1)},\;\; k\geq 1,\;\; s\geq k+1,
\label{56477872}
\end{equation}
where
\[
\hat{T}^k_s(n)=\sum_{j=0}^k c_j T^{k-j}_{s+j}(n),\;\;\; \forall s\geq 1.
\]
In \cite{Svinin3} it was shown that relations of the form (\ref{56477872}) play the role of constraints compatible with the Volterra chain and its higher symmetries.
\end{remark}
\begin{lemma}
Equation (\ref{87654999}) is equivalent to the associated recurrence relation (\ref{6766438}).
\end{lemma}
\begin{proof}
The number of parameters $\left(u_0,\ldots, u_{2g-1}; \alpha_0,\ldots, \alpha_g\right)$ needed to uniquely determine a sequence $\left(u_n\right)_{n\in \mathbb{Z}}$ by means of (\ref{6766438}) is $3g+1$. At the same time, the number of parameters $\left(u_0,\ldots, u_{2g}; c_1,\ldots, c_g\right)$ needed to uniquely determine a sequence defined by (\ref{87654999}) is also $3g+1$, which is a necessary condition for the equivalence of (\ref{6766438}) and (\ref{87654999}).
Furthermore, it is easy to see that
\[
K=\frac{\hat{T}^g_g(n+1)}{\prod_{j=0}^{2g} u_{n+j}}
\]
is an invariant for (\ref{87654999}). This relation, rewritten as
\begin{equation}
\prod_{j=0}^{2g} u_{n+j}=\frac{\hat{T}^g_g(n+1)}{K},
\label{8832}
\end{equation}
where $K$ remains an arbitrary constant, actually coincides with (\ref{6766438}). Indeed, comparing (\ref{8832}) with (\ref{6766438}) we obtain
\[
\alpha_j=\frac{c_{g-j}}{K},\;\;  j=0,\ldots, g.
\]
On the other hand, from (\ref{8832}) one can easily derive (\ref{6766438}) as a consequence.
\end{proof}
Thus, any $R_{2g+3}$-sequence gives some solution of equation (\ref{87654999}) with suitable parameters $\left(c_1,\ldots, c_g\right)$.

\subsection{Discrete polynomials $P^k_s(n)$}

Our immediate goal is to show how to construct a finite set of independent invariants for equation (\ref{87654999}), but, although this may not be obvious, we need to construct another class of discrete polynomials closely related to $T^k_s(n)$.

Recall that the explicit formula (\ref{67999998}) for the polynomials $T^k_s(n)$ is a sum of monomials over a finite set of index values $D_{k, s}=\left\{\lambda_j : 0\leq \lambda_1<\cdots< \lambda_k\leq s-1\right\}$. Consider the partition $D_{k, s}=\bar{D}_{k, s}\sqcup\tilde{D}_{k, s}$, where
\[
\bar{D}_{k, s}=\left\{\lambda_j : \lambda_1=0,\; \lambda_k=s-1,\; 1\leq \lambda_2<\cdots< \lambda_{k-1}\leq s-2\right\},
\]
and $\tilde{D}_{k, s}$ is its complement. Then the definition of the discrete polynomial $P^k_s(n)$ is as follows. The formula is the same as for $T^k_s(n)$, i.e. (\ref{67999998}), only the set $D_{k, s}$ is replaced by its subset $\tilde{D}_{k, s}$. In other words, in the summation we need to remove from $D_{k, s}$ all elements where $\lambda_1=0$ and $\lambda_k=s-1$. This requirement leads to the relation
\begin{equation}
P^k_s(n)=T^k_s(n)-u_nu_{n+k+s-2}T^{k-2}_{s-2}(n+2),\;\;\; \forall s\geq k\geq 2.
\label{77654987}
\end{equation}
Note that, since $\tilde{D}_{k, s}=\varnothing$ for $s=1,\ldots, k$, in this case we have $P^k_s(n)=0$.

The following lemma follows from two possible partitions of the set $\tilde{D}_{k, s}$, namely:
\begin{eqnarray*}
\tilde{D}_{k, s}&=&\left\{\lambda_j : 1\leq \lambda_1<\cdots< \lambda_{k}\leq s-1\right\} \\
&& \sqcup \left\{\lambda_j : \lambda_1=0,\;\; 1\leq \lambda_2<\cdots< \lambda_{k}\leq s-2\right\}
\end{eqnarray*}
and
\begin{eqnarray*}
\tilde{D}_{k, s}&=&\left\{\lambda_j : 0\leq \lambda_1<\cdots< \lambda_{k}\leq s-2\right\} \\
&& \sqcup \left\{\lambda_j : \lambda_k=s-1,\;\; 1\leq \lambda_1<\cdots< \lambda_{k-1}\leq s-2\right\}.
\end{eqnarray*}
\begin{lemma} 
The following two relations
\begin{eqnarray}
P^k_s(n)&=&T^k_{s-1}(n+1)+u_n T^{k-1}_{s-2}(n+2) \label{64265} \\
&=&T^k_{s-1}(n)+u_{n+k+s-2} T^{k-1}_{s-2}(n+1),
\label{642}
\end{eqnarray}
for all $s\geq 1$, are identities.
\end{lemma}
\begin{remark}
In formulas (\ref{64265}) and (\ref{642}) we should assume $T^k_s(n)=0,\;\; \forall s<k$. If $s=k$, then taking (\ref{6799999000988}) into account, both these relations give
\[
P^k_{k+1}(n)=\prod_{j=0}^{k-1}u_{n+2j} + \prod_{j=0}^{k-1}u_{n+2j+1},\;\; \forall k\geq 1.
\]
\end{remark}

\subsection{Invariants for equations of the form (\ref{87654999})}

Next we need the following technical statement.
\begin{lemma} \label{78654}
The relation
\begin{equation}
u_n\hat{T}^s_s(n+2)-u_{n+2s+1}\hat{T}^s_s(n+1)=\hat{T}^{s+1}_{s+1}(n)-\hat{T}^{s+1}_{s+1}(n+1)
\label{87677549}
\end{equation}
is an identity.
\end{lemma}
\begin{proof}
Since $c_j$ are arbitrary parameters, this means that (\ref{87677549}) must be a generating relation for a finite set of identities of the form
\[
u_nT^{s-j}_{s+j}(n+2)-u_{n+2s+1}T^{s-j}_{s+j}(n+1)=T^{s-j+1}_{s+j+1}(n)-T^{s-j+1}_{s+j+1}(n+1),\;\;\; j=0,\ldots, s.
\]
But these, as is easy to see, are a special case of the identity
\[
u_nT^{k-1}_{s-1}(n+2)-u_{n+s+k-1}T^{k-1}_{s-1}(n+1)=T^k_s(n)-T^k_s(n+1)
\]
which, in turn, follows from identities (\ref{6540987}) and (\ref{67776587}).
\end{proof}

For further computations we will denote the invariant $K$ somewhat more elaborately, namely as $K_g(n)$. Clearly, the relation
\[
K_g(n+1)-K_g(n)=\frac{1}{\prod_{j=0}^{2g+1} u_{n+j}}\left(u_n\hat{T}^g_g(n+2)-u_{n+2g+1}\hat{T}^g_g(n+1)\right)
\]
holds. By Lemma \ref{78654}, we can rewrite this relation as
\begin{equation}
K_g(n+1)-K_g(n)=\frac{1}{\prod_{j=0}^{2g+1} u_{n+j}}\left(\hat{T}^{g+1}_{g+1}(n)-\hat{T}^{g+1}_{g+1}(n+1)\right).
\label{870009}
\end{equation}

Consider the following set of functions:
\begin{equation}
\tilde{K}_{g, s}(n)=\sum_{j=0}^s P^j_{2(g-s+1)+j}(n+s-j) K_{g-s+j}(n+s-j),\;\;\; s=1,\ldots, g,
\label{8787659}
\end{equation}
where\footnote{We set $\hat{T}^0_0(n)=1$.}
\[
K_s(n)=\frac{\hat{T}^s_s(n+1)}{\prod_{j=0}^{2s} u_{n+j}},\;\; \forall s\geq 0.
\]
Note that $\tilde{K}_{g, s}(n)$ is a Laurent polynomial in the arguments $\left(u_n,\ldots, u_{n+2g}\right)$ for all $s=1,\ldots, g$. Direct computation verifies the following statement.
\begin{lemma}
The recurrence relation
\begin{equation}
\tilde{K}_{g, s}(n)=P^s_{2g-s+2}(n) K_g(n)+\tilde{K}_{g-1, s-1}(n+1).
\label{87666549}
\end{equation}
holds.
\end{lemma}
This lemma and Lemma \ref{78654} are auxiliary to the following statement.
\begin{lemma} \label{675432} 
The two relations
\begin{eqnarray}
\tilde{K}_{g, s}(n+1)-\tilde{K}_{g, s}(n)&=&\frac{T^s_{2g-s+1}(n+1)}{\prod_{j=0}^{2g+1}u_{n+j}}\left(u_n{\hat{T}^g_g(n+2)}-u_{n+2g+1}{\hat{T}^g_g(n+1)}\right)
\label{87691} \\
&=&\frac{T^s_{2g-s+1}(n+1)}{\prod_{j=0}^{2g+1}u_{n+j}}\left({\hat{T}^{g+1}_{g+1}(n)}-{\hat{T}^{g+1}_{g+1}(n+1)}\right),
\label{8769}
\end{eqnarray}
for any $g\geq 1$ and $s=1,\ldots, g$, are identities.
\end{lemma}
\begin{proof}
First of all, we note that the equality (\ref{87691})=(\ref{8769}) follows from Lemma \ref{78654}. Taking (\ref{87666549}) into account, we obtain
\begin{eqnarray*}
\tilde{K}_{g, s}(n+1)-\tilde{K}_{g, s}(n)&=&\tilde{K}_{g-1, s-1}(n+2)-\tilde{K}_{g-1, s-1}(n+1)  \\
&&+P^s_{2g-s+2}(n+1)K_g(n+1)-P^s_{2g-s+2}(n)K_g(n). \nonumber
\end{eqnarray*}
Next we use mathematical induction. Suppose we have proved relation (\ref{8769}) with the replacement $g\rightarrow g-1$, $s\rightarrow s-1$ and $n\rightarrow n+1$. Then we can perform the following computations:
\begin{eqnarray}
\tilde{K}_{g, s}(n+1)-\tilde{K}_{g, s}(n)&=&\frac{T^{s-1}_{2g-s}(n+2)}{\prod_{j=1}^{2g}u_{n+j}}\left(\hat{T}^{g}_{g}(n+1)-\hat{T}^{g}_{g}(n+2)\right) \nonumber\\
&&+P^s_{2g-s+2}(n+1)\frac{\hat{T}^g_{g}(n+2)}{\prod_{j=1}^{2g+1}u_{n+j}}-P^s_{2g-s+2}(n)\frac{\hat{T}^g_{g}(n+1)}{\prod_{j=0}^{2g}u_{n+j}} \nonumber\\
&=&\left(P^s_{2g-s+2}(n+1)-u_{n+2g+1}T^{s-1}_{2g-s}(n+2)\right)\frac{\hat{T}^g_{g}(n+2)}{\prod_{j=1}^{2g+1}u_{n+j}} \nonumber\\
&&-\left(P^s_{2g-s+2}(n)-u_{n}T^{s-1}_{2g-s}(n+2)\right)\frac{\hat{T}^g_{g}(n+1)}{\prod_{j=0}^{2g}u_{n+j}}.
\label{67543209}
\end{eqnarray}
Now it is appropriate to use the two identities
\begin{eqnarray*}
T^s_{2g-s+1}(n+1)&=&P^s_{2g-s+2}(n+1)-u_{n+2g+1}T^{s-1}_{2g-s}(n+2) \\
&=&P^s_{2g-s+2}(n)-u_{n}T^{s-1}_{2g-s}(n+2),
\end{eqnarray*}
which follow from (\ref{64265}) and (\ref{642}). As a result, we obtain
\[
\tilde{K}_{g, s}(n+1)-\tilde{K}_{g, s}(n)=\frac{T^s_{2g-s+1}(n+1)}{\prod_{j=0}^{2g+1}u_{n+j}}\left(u_{n}\hat{T}^g_g(n+2)-u_{n+2g+1}\hat{T}^g_g(n+1)\right).
\]
Thus, we have shown that the validity of (\ref{8769}) for $(g-1, s-1)$ with some values of $g$ and $s$ implies the validity of (\ref{87691}) for $(g, s)$. Now we use the equality (\ref{87691})=(\ref{8769}). Note that $\tilde{K}_{g, 0}(n)=K_{g}(n)$ and for this function relations (\ref{87691}) and (\ref{8769}) are obviously true for any $g\geq 1$. Formally, these relations also hold for $\tilde{K}_{0, 0}(n)=1/u_{n}$. The lemma is proved by mathematical induction.
\end{proof}
We formulate the obtained result as a theorem.
\begin{theorem} \label{675430}
Equation (\ref{87654999}) has a set of invariants $\left(\tilde{K}_{g, 0},\ldots, \tilde{K}_{g, g}\right)$, defined by (\ref{8787659}).
\end{theorem}
\begin{example}   \label{875432}
For the case $g=1$, equation (\ref{87654999}) takes the following form:
\begin{equation}
u_{n+3}=u_n\frac{u_{n+2}+c_1}{u_{n+1}+c_1}.
\label{7865498}
\end{equation}
According to Theorem \ref{675430} we have two invariants\footnote{We use here simpler notation: $\tilde{K}_s$ instead of $\tilde{K}_{g, s}$.}
\begin{equation}
\tilde{K}_{0}=K=\frac{u_{n+1}+c_1}{u_nu_{n+1}u_{n+2}}\;\; \mbox{and}\;\;  \tilde{K}_{1}=\frac{1}{u_{n+1}}+\left(u_n+u_{n+1}+u_{n+2}\right)\frac{u_{n+1}+c_1}{u_nu_{n+1}u_{n+2}}.
\label{7865490008}
\end{equation}
We solve the first relation for $u_{n+2}$. As a result, we obtain
\[
u_{n+2}=\frac{u_{n+1}+c_1}{K u_nu_{n+1}}.
\]
Substituting this expression into the second relation in (\ref{7865490008}), we obtain
\[
\frac{\tilde{K}_1}{K}=u_n+u_{n+1}+\frac{\alpha_0+\alpha_1(u_n+u_{n+1})}{u_nu_{n+1}}
\]
where $\alpha_0=c_1/K,\; \alpha_1=1/K$. As a result, we have obtained an invariant for the associated equation (\ref{6766438}) in the case $g=1$.
\end{example}
\begin{example}  \label{7777}
For the case $g=2$, equation (\ref{87654999}) takes the following form:
\[
u_{n+5}=u_n\frac{u_{n+2}u_{n+4}+c_1\left(u_{n+2}+u_{n+3}+u_{n+4}\right)+c_2}{u_{n+1}u_{n+3}+c_1\left(u_{n+1}+u_{n+2}+u_{n+3}\right)+c_2}.
\]
In this case, we have the following three invariants:
\[
\tilde{K}_{0}=K=\frac{u_{n+1}u_{n+3}+c_1\left(u_{n+1}+u_{n+2}+u_{n+3}\right)+c_2}{u_nu_{n+1}u_{n+2}u_{n+3}u_{n+4}},
\]
\begin{eqnarray*}
\tilde{K}_{ 1}&=&\frac{u_{n+2}+c_1}{u_{n+1}u_{n+2}u_{n+3}} +\left(u_n+u_{n+1}+u_{n+2}+u_{n+3}+u_{n+4}\right) \\
&&\times \frac{u_{n+1}u_{n+3}+c_1\left(u_{n+1}+u_{n+2}+u_{n+3}\right)+c_2}{u_nu_{n+1}u_{n+2}u_{n+3}u_{n+4}}
\end{eqnarray*}
and
\begin{eqnarray*}
\tilde{K}_{2}&=&\frac{1}{u_{n+2}}+\left(u_{n+1}+u_{n+2}+u_{n+3}\right)\frac{u_{n+2}+c_1}{u_{n+1}u_{n+2}u_{n+3}} \\
&&+\left(u_n(u_{n+2}+u_{n+3})+u_{n+1}(u_{n+3}+u_{n+4}\right)+u_{n+2}u_{n+4}) \\
&&\times \frac{u_{n+1}u_{n+3}+c_1\left(u_{n+1}+u_{n+2}+u_{n+3}\right)+c_2}{u_nu_{n+1}u_{n+2}u_{n+3}u_{n+4}}.
\end{eqnarray*}
We perform similar computations as in the previous example to derive invariants for the associated equation (\ref{6766438}) in the case $g=2$. As a result, we obtain
\[
\frac{\tilde{K}_1}{K}=u_n+u_{n+1}+u_{n+2}+u_{n+3}+\frac{\alpha_0+\alpha_1(u_n+u_{n+1}+u_{n+2}+u_{n+3})+\alpha_2(u_nu_{n+2}+u_{n+1}u_{n+3})}{u_nu_{n+1}u_{n+2}u_{n+3}}
\]
and
\begin{eqnarray*}
\frac{\tilde{K}_2}{K}&=&u_n\left(u_{n+2}+u_{n+3}\right)+u_{n+1}u_{n+3}+\left(u_{n+1}+u_{n+2}\right) \\
&&\times \frac{\alpha_0+\alpha_1(u_n+u_{n+1}+u_{n+2}+u_{n+3})+\alpha_2(u_nu_{n+2}+u_{n+1}u_{n+3})}{u_nu_{n+1}u_{n+2}u_{n+3}}  \\
&&+\frac{\alpha_1+\alpha_2(u_{n+1}+u_{n+2})}{u_{n+1}u_{n+2}}, 
\end{eqnarray*}
where $\alpha_0=c_2/K,\; \alpha_1=c_1/K,\; \alpha_2=1/K $.
\end{example}

\subsection{Invariants for equation (\ref{6766438})}

The computations in the above examples help to guess how the invariants $\tilde{K}_{g, s}(n)$ behave when we pass from the recurrence relation (\ref{87654999}) to its equivalent (\ref{6766438}). To answer this question, we define\footnote{In (\ref{89743209}) we use the notation 
\[
\hat{P}^s_{s+1}(n)=\sum_{j=0}^s c_j P^{s-j}_{s+j+1}(n).
\] 
}
\begin{equation}
J_{g, s}(n)= K \cdot T^{s}_{2g-s+1}(n) +\sum_{j=0}^{s-1} T^{s-j-1}_{2g-s-j}(n+j+1)\frac{\hat{P}^{g-j}_{g-j+1}(n+j)}{\prod_{l=j}^{2g-j-1} u_{n+l}},\;\;\forall s=1,\ldots, g.
\label{89743209}
\end{equation}
Note that (\ref{89743209}) is a function of the variables $\left(u_n,\ldots, u_{n+2g-1}\right)$.
\begin{lemma} \label{674320976}
The relation
\begin{equation}
J_{g, s}(n)-\tilde{K}_{g, s}(n)=T^s_{2g-s+1}(n)\left(K-\frac{\hat{T}^g_g(n+1)}{\prod_{j=0}^{2g}u_{n+j}}\right),\;\;\; \forall s=1,\ldots, g,
\label{8974320009999}
\end{equation}
is an identity.
\end{lemma}
\begin{proof}
Substituting (\ref{8787659}) and (\ref{89743209}) into (\ref{8974320009999}), we obtain a rather cumbersome relation:
\[
\sum_{j=0}^{s-1} T^{s-j-1}_{2g-s-j}(n+j+1)\frac{\hat{P}^{g-j}_{g-j+1}(n+j)}{\prod_{l=j}^{2g-j-1} u_{n+l}}-\sum_{j=0}^s P^{s-j}_{2g-s-j+2}(n+j) \frac{\hat{T}^{g-j}_{g-j}(n+j+1)}{\prod_{l=j}^{2g-j} u_{n+l}}
\]
\begin{equation}
=-T^s_{2g-s+1}(n)\frac{\hat{T}^g_g(n+1)}{\prod_{j=0}^{2g} u_{n+j}}.
\label{7643209}
\end{equation}

We need to prove that this is an identity. First, consider the simplest case $s=1$. It is not difficult to show that in this case we obtain the relation
\begin{equation}
\hat{P}^g_{g+1}(n)-\hat{T}^g_{g}(n+1)-u_n \hat{T}^{g-1}_{g-1}(n+2)=0.
\label{7645999879}
\end{equation}
This is simply the generating relation for a finite set of identities
\[
P^{g-s}_{g+s+1}(n)=T^{g-s}_{g+s}(n+1)+u_n T^{g-s-1}_{g+s-1}(n+2),\;\; \forall s=0,\ldots, g-1
\]
each of which follows from (\ref{64265}).

Next we perform the following two steps. First step: using the identity
\[
P^s_{2g-s+2}(n)-T^s_{2g-s+1}(n)=u_{n+2g}T^{s-1}_{2g-s}(n+1),
\]
which follows from (\ref{642}), we rewrite relation (\ref{7643209}) in equivalent form as\footnote{To avoid rewriting the cumbersome expressions in the two sums of (\ref{7643209}), we denote them by $\left(\cdot\right)$.}
\[
\sum_{j=0}^{s-1} \left(\cdot\right)-\sum_{j=1}^s \left(\cdot\right)=T^{s-1}_{2g-s}(n+1)\frac{\hat{T}^g_g(n+1)}{\prod_{j=0}^{2g-1} u_{n+j}}.
\]
Second step: using identity (\ref{7645999879}), we rewrite the last relation as follows:
\begin{equation}
\sum_{j=1}^{s-1} \left(\cdot\right)-\sum_{j=1}^s \left(\cdot\right)=-T^{s-1}_{2g-s}(n+1)\frac{\hat{T}^{g-1}_{g-1}(n+2)}{\prod_{j=1}^{2g-1} u_{n+j}}.
\label{7677777}
\end{equation}
As a result, we have shown that relation (\ref{7643209}), which we need to prove, is equivalent to (\ref{7677777}). How do these two relations differ? First, the right-hand side of (\ref{7677777}) can be obtained as follows: take the right-hand side of (\ref{7643209}) and replace $g\rightarrow g-1$, $s\rightarrow s-1$, $n\rightarrow n+1$. Second, one can notice that the summations on the left-hand side start not from zero but from one. Now we perform these two steps of computations several times, or more precisely, their analogues. As a result, we obtain that the required identity (\ref{7643209}) is equivalent to
\[
\sum_{j=k}^{s-1} \left(\cdot\right)-\sum_{j=k}^s \left(\cdot\right)=-T^{s-k}_{2g-s-k+1}(n+k)\frac{\hat{T}^{g-k}_{g-k}(n+k+1)}{\prod_{j=k}^{2g-k} u_{n+j}},
\]
where $k\leq s$. Setting $k=s$ here, this relation reduces to a trivial identity.
\end{proof}

\subsection{Main result}

Having proved this lemma, we obtain a finite set of invariants for the recurrence relation (\ref{8832}) and for (\ref{6766438}). We will write $H_s$ instead of $J_{g, s}/K$. Recall that $J_{g, s}$ is given by formula (\ref{89743209}). We can adapt this expression for our associated equation (\ref{6766438}) to write its invariants, say as in the following theorem, which we have just proved by combining Lemmas \ref{675432} and \ref{674320976}.
\begin{theorem} \label{675422212}
The recurrence relation (\ref{6766438}) has the following invariants:
\begin{equation}
H_{s}=T^{s}_{2g-s+1}(n)+\sum_{j=0}^{s-1} T^{s-j-1}_{2g-s-j}(n+j+1)\frac{P_{j, g-j}(n+j)}{\prod_{l=j}^{2g-j-1} u_{n+l}},\;\; s=1,\ldots, g,
\label{7688767}
\end{equation}
where
\[
P_{j, g-j}(n)=\sum_{s=0}^{g-j} \alpha_{j+s}P^{s}_{2g-2j-s+1}(n),\;\; j=0,\ldots, g.
\]
\end{theorem}
Thus, the main result of this section is this theorem. We have constructed a finite set of invariants for the associated recurrence relation (\ref{6766438}), expressed in terms of the discrete polynomials $T^k_s(n)$. Thereby, we have also obtained a finite set of invariants for $R_{2g+3}$.

We can transform relation (\ref{7688767}). Let
\begin{equation}
T_{j, g-j}(n)=\sum_{s=0}^{g-j} \alpha_{j+s}T^{s}_{2g-2j-s}(n),\;\; j=0,\ldots, g.
\label{000767}
\end{equation}
Direct verification proves the following statement.
\begin{lemma}
In view of (\ref{642}), the identity
\[
P_{j, g-j}(n)=T_{j, g-j}(n)+u_{n+2g-2j-1}T_{j+1, g-j-1}(n+1),
\]
holds for any $j\geq 0$.
\end{lemma}
Using this lemma and applying identity (\ref{64265}), we can rewrite the expression for the invariants (\ref{7688767}) as
\begin{equation}
H_{s}=T^{s}_{2g-s+1}(n)+T^{s-1}_{2g-s}(n+1)\frac{T_{0, g}(n)}{\prod_{j=0}^{2g-1} u_{n+j}}+\sum_{j=1}^s P^{s-j}_{2g-s+j}(n) \frac{T_{j, g-j}(n+j)}{\prod_{l=j-1}^{2g-j-1} u_{n+l}}.
\label{0776547}
\end{equation}
A nontrivial fact is that all these invariants depend linearly on the parameters of the equation (\ref{6766438}). Moreover, one can notice that they are expressed through standard linear combinations $T_{j, g-j}(n)$, defined by (\ref{000767}). Note that the dependence between these invariants could only be linear, but it is fairly obvious that such dependence cannot exist, if only because the \textit{free} term in $H_s$ is a discrete homogeneous polynomial of degree $s$.

\section{More general criteria for integrality} \label{765498}

 We present an equation which is conjecturally equivalent to relation (\ref{6766438}). It is defined in terms of certain discrete polynomials $S^k_s(n)$, which are closely related to $T^k_s(n)$.

\subsection{Discrete polynomials $S^k_s(n)$}

In addition to $T^k_s(n)$, we define discrete polynomials $S^k_s(n)$ by the relation \cite{Svinin2}
\[
S^k_s(n)=\sum_{j=0}^{s-1} u_{n+j+k-1} S^{k-1}_{s-j}(n+j)
\]
with the condition $S^0_s(n)=1$. An explicit formula can also be written for these polynomials:
\[
S^k_s(n)=\sum_{0\leq \lambda_1\leq \cdots\leq  \lambda_k\leq s-1} \prod_{j=1}^k u_{n+\lambda_j+k-j}. 
\]

\subsection{Equivalent equation}

\begin{conjecture}  \label{675438888209}
The associated recurrence relation (\ref{6766438}) is equivalent to the equation
\begin{equation}
u_{n+g}\hat{S}^g_{g+2}(n)=(-1)^g\alpha_g,
\label{67543298}
\end{equation}
where 
\[
\hat{S}^g_{g+2}(n)=\sum_{j=0}^{g} (-1)^j  S^{g-j}_{g-j+2}(n+j)H_j.
\]
The coefficients $H_j$ here are the invariants from Theorem \ref{675422212}.
\end{conjecture}
Relation (\ref{67543298}) is an equation of order $2g$ with parameters. Its solution is determined by the data $\left(u_0,\ldots, u_{2g-1}; H_1,\ldots, H_g; \alpha_g\right)\in \mathbb{C}^{3g+1}$. Despite the fact that Conjecture \ref{675438888209} is not proved in general, for small $g$ this statement can be verified directly by substituting the expressions (\ref{0776547}) for the invariants. Below we consider two cases of computations related to equation (\ref{67543298}) for $g=1$ and $g=2$. We emphasize once again that in these specific cases the above statement is not conjectural but can be verified directly.

\subsection{$g=1$. Equivalent equation}

In this case, relation (\ref{67543298}) is written as
\begin{equation}
u_{n+1}\left(u_n+u_{n+1}+u_{n+2}-H_1\right)+\alpha_1=0.
\label{7863257}
\end{equation}
We compute the invariant $H_1$ by formula (\ref{0776547}) and obtain\footnote{In essence, we have already given the expression for this invariant in Example \ref{875432}.}
\begin{equation}
H_1=u_n+u_{n+1}+\frac{\alpha_0+\alpha_1 u_n}{u_nu_{n+1}}+\frac{\alpha_1}{u_n}.
\label{78632000957}
\end{equation}
On one hand, substituting this expression into (\ref{7863257}), we obtain a special case of the associated equation (\ref{6766438}), namely
\begin{equation}
u_nu_{n+1}u_{n+2}=\alpha_0+\alpha_1 u_{n+1}.
\label{786325000987}
\end{equation}
On the other hand, one can check that
\begin{equation}
\alpha_0= -u_nu_{n+1}(u_n+u_{n+1}-H_1)-(u_n+u_{n+1})\alpha_1
\label{7887}
\end{equation}
is an invariant for equation (\ref{7863257}). Substituting (\ref{7887}) into (\ref{786325000987}) reduces this relation to equation (\ref{7863257}). Thus, equations (\ref{7863257}) and (\ref{786325000987}) are equivalent.
\begin{remark}
Relation (\ref{78632000957}) implies that $\left(u_n, u_{n+1}\right)\in \mathbb{C}^2$ lies on the curve defined by the equation
\[
\left(X+Y-H_1\right)\left(XY+\alpha_1\right)+\alpha_0+\alpha_1 H_1=0.
\]
It turns out that this equation defines an elliptic curve. By a fractional-linear change of variables, it can be reduced to the Weierstrass canonical form: $y^2=4x^3-g_2 x -g_3$. This fact allowed the solution of the associated equation (\ref{786325000987}) to be expressed in terms of the Weierstrass $\sigma$-function in \cite{Hone1}. This, in turn, made it possible to express any Somos-5 sequence in terms of the Weierstrass $\sigma$-function.
\end{remark}

\subsection{$g=1$. Construction of an integer Somos-5 sequence}

It is easy to check that, by means of substitution (\ref{676678888768}), relation (\ref{7863257}) reduces to the equation
\begin{equation}
t_nt_{n+3}^2t_{n+4}+t_{n+1}^2t_{n+4}^2+t_{n+1}t_{n+2}^2t_{n+5}-t_{n+1}t_{n+2}t_{n+3}t_{n+4} H_1+\alpha_1  t^2_{n+2}t^2_{n+3}=0.
\label{7863555555257}
\end{equation}
One can also check that substituting expression (\ref{78632000957})\footnote{Provided we have made the substitution (\ref{676678888768}).} into (\ref{7863555555257}) leads to the standard Somos-5 form (\ref{5}). Although relation (\ref{7863555555257}) looks more cumbersome than (\ref{5}), it has its own features. The following statement was proved in \cite{Svinin8}.
\begin{proposition} \label{67543224}
For equation (\ref{7863555555257}) we have
\[
t_n\in \mathbb{Z}[t_0, t_1^{\pm 1}, t_2^{\pm 1}, t_3^{\pm 1}, t_4; H_1, \alpha_1],\; \forall n\in\mathbb{Z}.
\]
\end{proposition}
Since $H_1\in \mathbb{Z}[t_0^{\pm 1}, t_1^{\pm 1}, t_2^{\pm 1}, t_3^{\pm 1}, t_4^{\pm 1}; \alpha_0, \alpha_1]$\footnote{Here we assume that $H_1$ is computed for $n=0$, which is of course not a principal issue.}, this means that the ordinary Laurent property for Somos-5 follows from Proposition \ref{67543224}, but not conversely.

From Proposition \ref{67543224} follows a method for constructing an integer Somos-5 sequence. We choose arbitrarily the following data:
\[
t_k=\pm 1\;\; \mbox{for}\;\;  k=1, 2, 3\;\; \mbox{and}\;\;   \left(t_0,\; t_4,\; \alpha_1,\; H_1\right)\in\mathbb{Z}^4.
\] 
The value of the parameter $\alpha_0$ is computed by formula (\ref{7887}). Any such choice of data leads to some integer Somos-5 sequence. In particular, the data
\[
t_k=1\;\; \mbox{for}\;\;  k=1, 2, 3\;\; \mbox{and}\;\;   \left(t_0,\; t_4,\; \alpha_1,\; H_1\right)=\left(1,\; 1,\; 1,\; 5\right)
\] 
correspond to the sequence \textrm{A006721}.
\begin{remark}
For any chosen data, we obtain $\left(u_0,\; u_1\right) \in\mathbb{Z}^2$ and consequently, in view of (\ref{7887}), $\alpha_0\in\mathbb{Z}$, although we do not, generally speaking, require $\alpha_0$ to be integer.
\end{remark}

\subsection{$g=2$. Equivalent equation.}

In this case, equation (\ref{67543298}) in explicit form looks as follows:
\begin{eqnarray}
&& u_{n+2}\left(u_{n+1}(u_n+u_{n+1}+u_{n+2}+u_{n+3})+u_{n+2}(u_{n+1}+u_{n+2}+u_{n+3}) \right. \nonumber\\
&&\;\;\;\;\;\;\;\;\;\; \left.  +u_{n+3}(u_{n+2}+u_{n+3})+u_{n+4}u_{n+3}-H_1(u_{n+1}+u_{n+2}+u_{n+3})+H_2\right) \nonumber\\
&&\;\;\;\;\;\;\;\;\;\; -\alpha_2=0.
\label{79}
\end{eqnarray}
\begin{remark}
It turns out that (\ref{79}), up to notation, is equation~1.2 in \cite{Hone2}; however, it originally appeared in \cite{Gubbiotti} in the classification of four-dimensional maps with two independent invariants. In \cite{Hone2} it was shown that this equation is Liouville integrable. There one can also find a visualization of ten thousand points $\left(u_n, u_{n+1}, u_{n+2}\right)\in\mathbb{Q}^3$, generated by the recurrence relation (\ref{79}) with the values $\left(H_1, H_2, \alpha_2\right)=\left(10, 29, 9\right)$ and rational positive initial conditions.
\end{remark}
It can be shown that (\ref{79}) is equivalent to the associated equation
\begin{equation}
u_nu_{n+1}u_{n+2}u_{n+3}u_{n+4}=\alpha_0+\alpha_1 \left(u_{n+1}+u_{n+2}+u_{n+3}\right)+\alpha_2  u_{n+1}u_{n+3}.
\label{77654}
\end{equation}
Indeed, one can show that equation (\ref{79}) has two invariants
\begin{eqnarray}
\alpha_0&=& -u_{n+1}u_{n+2}\left\{\left(u_{n+1}(u_{n}+u_{n+1}+u_{n+2})+u_{n+2}(u_{n+1}+u_{n+2})+u_{n+3}u_{n+2} \right) \right. \nonumber\\
&& \times (u_{n}+u_{n+1}+u_{n+2}+u_{n+3}) \nonumber\\
&&-(u_{n}+u_{n+1}+u_{n+2})(u_{n+1}+u_{n+2}+u_{n+3}) H_1 \nonumber\\
&&\left. +(u_{n}+u_{n+1}+u_{n+2}+u_{n+3})H_2\right\} \nonumber\\
&&+(u_{n+1}(u_{n}+u_{n+1}+u_{n+2})+u_{n+2}(u_{n+1}+u_{n+2})+u_{n+3}u_{n+2})\alpha_2,
\label{67549876}
\end{eqnarray}
\begin{eqnarray}
\alpha_1&=& u_{n+1}u_{n+2}\left\{u_{n+1}(u_{n}+u_{n+1}+u_{n+2})+u_{n+2}(u_{n+1}+u_{n+2})+u_{n+3}u_{n+2}-u_nu_{n+3} \right.\nonumber \\
&&\left.  -(u_{n+1}+u_{n+2})H_1+H_2\right\}-(u_{n+1}+u_{n+2})\alpha_2
\label{677776656}
\end{eqnarray}
and their substitution into (\ref{77654}) leads to (\ref{79}). On the other hand, substituting $H_1$ and $H_2$\footnote{The expressions for the invariants $H_1$ and $H_2$ are given in Example \ref{7777}.} into (\ref{79}) yields (\ref{77654}).

\subsection{$g=2$. Construction of an integer $R_7$ sequence}

Substitution of (\ref{676678888768}) into (\ref{79}) leads to the following relation:
\begin{eqnarray}
&&t_nt_{n+3}^2t_{n+4}^3t_{n+5}^2+t_{n+1}^2t_{n+4}^4t_{n+5}^2+2 t_{n+1} t_{n+2}^2 t_{n+4}^2t_{n+5}^3+t_{n+1}t_{n+2}t_{n+3}^2 t_{n+4}^2 t_{n+5}t_{n+6}  \nonumber\\
&&\;\;\;\;\;\;\;\;\;\;\;\;\; +t_{n+2}^4t_{n+5}^4+2t_{n+2}^3t_{n+5}^2t_{n+3}^2t_{n+6}+t_{n+3}^4t_{n+6}^2t_{n+2}^2+t_{n+7}t_{n+3}^3t_{n+2}^2t_{n+4}^2 \nonumber \\
&&\;\;\;\;\;\;\;\;\;\;\;\;\;  -t_{n+2}t_{n+3}t_{n+4}t_{n+5}\left(t_{n+1}t_{n+4}^2t_{n+5}+t_{n+5}^2t_{n+2}^2+t_{n+6}t_{n+3}^2t_{n+2}\right)H_1 \nonumber \\
&&\;\;\;\;\;\;\;\;\;\;\;\;\; +t_{n+2}^2t_{n+3}^2t_{n+4}^2t_{n+5}^2H_2-t_{n+2}t_{n+3}^3t_{n+4}^3t_{n+5} \alpha_2 =0. 
\label{786453209}
\end{eqnarray}
It was also written out, albeit in different notation, in \cite{Hone2}. More precisely, this is relation~2.4 in that paper. One can verify that substituting $H_1$ and $H_2$ into (\ref{786453209}) reduces it to the equation $R_7$. The following statement was proved in \cite{Hone2}.
\begin{proposition} \label{644324}
For equation (\ref{786453209}) we have
\[
t_n\in \mathbb{Z}[t_0, t_1, t_2^{\pm 1}, t_3^{\pm 1}, t_4^{\pm 1}, t_5, t_6; H_1, H_2, \alpha_2],\; \forall n\in\mathbb{Z}.
\]
\end{proposition}
From this proposition follows the following algorithm for constructing an integer $R_7$ sequence. We choose arbitrarily the following data:
\[
t_k=\pm 1\;\; \mbox{for}\;\;  k=2, 3, 4\;\; \mbox{and}\;\;   \left(t_0, t_1, t_5, t_6;  H_1, H_2, \alpha_2\right)\in\mathbb{Z}^7.
\] 
Using formulas (\ref{67549876}) and (\ref{677776656}) we compute the values of the parameters $\alpha_0$ and $\alpha_1$.
\begin{remark}
The proof of Propositions \ref{67543224} and \ref{644324} is of the same type. The idea of the proof is contained in \cite{Hone2}. For each $R_{2g+3}$ there exists an equivalent equation of type (\ref{7863555555257}) or (\ref{786453209}). In turn, for each such equivalent equation a statement similar to Propositions \ref{67543224} or \ref{644324} holds.
\end{remark}
\begin{remark}
Examples of equations $R_{2g+3}$, namely $R_7$ and $R_9$, were written out in \cite{Svinin8}. At the end of that paper we conjectured the existence of an infinite family of such equations possessing the Laurent property. In \cite{Svinin7} it was shown how to construct this family and the Laurent property was proved for these equations. There we also defined an infinite family of even-order equations which should also have the Laurent property. Computations using computer algebra confirm, although do not prove, this fact.
\end{remark}

\section*{Acknowledgements}

This work was carried out within the framework of  the state assignment of the Ministry of Education
and Science of the Russian Federation on the project No.  126021217175-3.

\end{document}